\documentclass[11pt]{article}
\usepackage{odonnell}
\usepackage{stmaryrd}
\usepackage{qcircuit}
\usepackage{physics}
\usepackage{mathtools}

\newcommand{\mc}{\mathcal}
\DeclareMathOperator\eye{\mathbb{I}}
\usepackage[linesnumbered,ruled,vlined]{algorithm2e}
\usepackage[backend=bibtex8,style=alphabetic,doi=false,isbn=false,url=false,maxbibnames=5,sorting=nty,sortcites=false]{biblatex}
\usepackage[affil-it]{authblk}
\renewcommand{\thefootnote}{\fnsymbol{footnote}}
\bibliography{bibliography,doms_refs}
\begin{document}

\title{Non-Clifford Gates are Required for Long-Term  Memory}
\author[1]{Jon Nelson$^{*\dagger}$}
\author[1,2]{Joel Rajakumar$^{*\ddagger}$}
\author[1]{Michael J. Gullans}
\affil[1]{\normalsize  Joint Center for Quantum Information \& Computer Science, University of Maryland and NIST}
\affil{\normalsize Department of Computer Science,
	University of Maryland}
\affil[2]{\normalsize  IBM T. J. Watson Research Center, Yorktown Heights, NY}
\date{}
\maketitle
\renewcommand{\thefootnote}{\fnsymbol{footnote}}
\footnotetext[1]{These authors contributed equally to this work.}
\renewcommand{\thefootnote}{\arabic{footnote}}
\renewcommand{\thefootnote}{\fnsymbol{footnote}}
\footnotetext[2]{nelson1@umd.edu}
\renewcommand{\thefootnote}{\arabic{footnote}}
\renewcommand{\thefootnote}{\fnsymbol{footnote}}
\footnotetext[3]{jrajakum@umd.edu}
\renewcommand{\thefootnote}{\arabic{footnote}}
\begin{abstract}
We show that all Clifford circuits under interspersed depolarizing noise lose memory of their input exponentially quickly, even when given access to a constant supply of fresh qubits in arbitrary states. This is somewhat surprising given the result of \cite{aharonov_fault-tolerant_1997} which gives a fault-tolerant protocol for general quantum circuits using a supply of fresh qubits. Our result shows that such a protocol is impossible using only Clifford gates demonstrating that non-Clifford gates are fundamentally required to store information for long periods of time.
\end{abstract}
\section{Introduction}
It has been known for some time that noisy quantum circuits can demonstrate fault-tolerance when given access to a supply of fresh qubits (or equivalently reset gates) \cite{aharonov_fault-tolerant_1997}. On the other hand, without this resource, it has been shown that \text{all} noisy circuits lose memory of their input exponentially quickly due to the build up of entropy \cite{M_ller_Hermes_2016,aharonov1996limitationsnoisyreversiblecomputation}. Here, we show that the non-Clifford gate, i.e. the resource of `magic', holds a similar power in terms of fault-tolerance. In particular, removing the non-Clifford gates from a noisy quantum (or classical) computation also causes it to lose memory of its input exponentially quickly, even when entropy can be expelled through reset gates.

We consider an identical setup to the threshold theorem proved in \cite{aharonov_fault-tolerant_1997}, where intermediate measurements are not allowed during the circuit\footnote{Note that intermediate measurements can be coherently simulated using the deferred measurement principle, but any intermediate classical processing on the measurement outcome must also be noisy and Clifford.}, however, the quantum circuit is allowed to reset any qubit, at any time, to any state. This setting is practically relevant as it captures the limitations of near-term quantum devices, where intermediate measurements remain challenging \cite{bluvstein2023logical}. These restrictions also address the question of whether noisy quantum circuits can protect quantum information on their own without the intervention of a noiseless external observer. This question has similarly motivated the field of self-correcting quantum memory \cite{Dennis_2002,alicki2008thermalstabilitytopologicalqubit}, and is a primary motivation in several fault-tolerance results \cite{aharonov1996limitationsnoisyreversiblecomputation,aharonov_fault-tolerant_1997,anshu_liu_nguyen_pattison_loglogn_iqp}. Furthermore, this setting allows us to prove that any Clifford protocol loses memory of $\textit{all}$ input states. Notice that this is trivially impossible when given intermediate measurements since one could simply measure any classical input state and store this information indefinitely in noiseless classical memory. Thus, we explicitly disallow noiseless classical resources.
 
Several recent results have demonstrated a similar loss of memory in noisy \textit{random} quantum circuits that contain resets and other non-unital operations \cite{mele2024noiseinducedshallowcircuitsabsence,angrisani2025simulatingquantumcircuitsarbitrary,martinez_simulation, schuster2024polynomialtimeclassicalalgorithmnoisy}. However, these results differ from ours in a primary aspect, even apart from the assumption of randomness: they show that the input state has negligible effect on the expected value of any \textit{fixed observable} after high enough depth, whereas we show that the input state has negligible effect on all possible observables at once. This is because we prove a convergence in \textit{trace distance}, which, to our knowledge, has never been shown before in a setting that includes non-unital operations and noise of any constant rate. Thus, we anticipate that our analysis techniques may be of broad interest. We state our main theorem below,
\begin{theorem} \label{theorem:main}
    Let $\Phi$ denote the channel produced by a circuit of $d$ layers of Clifford gates and reset gates on $n$ qubits, with single-qubit depolarizing noise of strength $\gamma$ on each qubit after each layer. Let $\rho$ and $\sigma$ be two arbitrary input states. When $d > d^*$ where \mbox{$d^* = O(\gamma^{-1}\log(n) \log (n/\epsilon))$},
    \begin{align}
        \|\Phi(\rho) - \Phi(\sigma)\|_1 \leq \epsilon
    \end{align}
\end{theorem}
Operationally, our result implies that any two input states become indistinguishable after polylogarithmically many layers of the circuit in the sense that any distinguisher correctly guesses which state it possesses with probability at most $1/2 +o(1/\poly(n))$.  

Besides the assumption that the circuit contains only Clifford gates and qubit resets, our result makes no additional assumptions on the circuit. In other words, the statement holds for \textit{any} Clifford circuit using mid-circuit resets in \textit{any} desired locations. In fact, the gates are not required to be geometrically local or even $O(1)$-local. The only requirement is that after each layer of resets and gates a round of constant-strength single-qubit depolarizing channels is applied to each qubit. Finally, the convergence in trace distance occurs for \textit{any} two quantum states. Despite the generality of the result, our proof is relatively short, relying primarily on a counting argument.

\subsection{Implications}
Our result is somewhat surprising since many error-correction gadgets can be constructed using Clifford gates. For example, one of the most popular types of error correcting codes is a stabilizer code which uses Clifford circuits for both its fault-tolerant encoding and fault-tolerant syndrome extraction protocols \cite{gottesman2013fault}. Clifford circuits are also desirable since they propagate Pauli errors nicely making it easier to track errors throughout the circuit, a property which has led to several fault-tolerant implementations of noisy logical Clifford circuits \cite{delfosse2023spacetime,bravyi_gosset_koenig_tomamichael_noisy, bacon_flammia_harrow_shi_subsystem_code, bergamaschi_liu}. 
Although these properties seem useful for correcting errors, our results point out that they are not so powerful on their own, and one must introduce a non-Clifford gate at some point of the computation in order to achieve long-term memory. For example, in order to make intermediate syndrome measurements robust to errors it is often required to perform a majority vote, which cannot be done using Clifford gates. In this way, our result highlights the necessary ingredients for effective fault-tolerant protocols.

Another reason our result is surprising is that magic is often discussed as a computational resource rather than as an error-correction resource. This is because Clifford circuits can be classically simulated efficiently \cite{gottesman1998heisenbergrepresentationquantumcomputers} even when augmented with a small number of non-Clifford gates \cite{Aaronson_2004,Bravyi_2019}. Thus, a high degree of non-Cliffordness is essential for achieving quantum advantage. In this work, we do not consider the problem of computational hardness but instead focus on the ability to preserve quantum states in memory. This is quite a different task and on the surface does not seem to require departure from Cliffordness as a necessary ingredient. However, our result shows that non-Clifford gates are crucial for error correction in addition to their importance in quantum advantage. 

A final feature of our result is that the intermediate reset gates are allowed to produce fresh magic states. One might think that this would allow for the implementation of non-Clifford gates through magic state injection. However, this technique requires a classically controlled Clifford operation as the final step, which is itself a non-Clifford gate. Thus, it is not clear how to use fresh magic states to overcome the limitations of the Clifford circuit. In fact, our no-go result shows that Clifford gates alone cannot exploit the resource of magic states for fault-tolerance, highlighting a distinction between the utility of magic states and non-Clifford gates. This contrasts with the case of computational hardness, where the inclusion of magic states is sufficient to promote Clifford circuits to a model that is hard to sample from classically \cite{bouland,Yoganathan_2019}.

\subsection{Proof Strategy}
Many related noisy circuit problems can be formulated in the Pauli path framework. Broadly speaking, a Pauli path is a trajectory in the Feynman path integral of the circuit with respect to the Pauli basis. In our case, to prove the loss of memory it is sufficient to show that the total contribution of each Pauli path connecting the input to the output state is negligible. To do this, it will be crucial to leverage the property that depolarizing noise causes the contribution of a Pauli path to decay exponentially in the number of non-identities in the path, also called the `weight'. This property is the key to many prior works that have also studied noise dynamics in quantum circuits \cite{Aharonov_2023,Bremner_2017,schuster2024polynomialtimeclassicalalgorithmnoisy,gao2018efficientclassicalsimulationnoisy}.

However, reset gates present a challenge for directly adapting these techniques as they introduce new Pauli paths throughout the circuit. This directly thwarts the Pauli path decay approach, which attempts to suppress the contribution of these paths. 
To apply Pauli path techniques to circuits with non-unital channels, we use a similar technique to \cite{mele2024noiseinducedshallowcircuitsabsence,schuster2024polynomialtimeclassicalalgorithmnoisy,angrisani2025simulatingquantumcircuitsarbitrary,martinez_simulation}, which considers the adjoint map of the noisy circuit. The reason for this is that the adjoint map has the nice property that it is once again unital, and thus the Pauli path decay techniques can be employed on the adjoint map. Unitalness is important for these decay techniques since by definition it means that the map has identity as a fixed point and the presence of depolarizing noise causes these paths to decay until the fixed point is reached. 

Even after reformulating our problem as a Pauli path decay problem, the difficulty of the reset gate still presents itself in other subtle ways that require careful analysis. In particular, the adjoint map of the $\ket{0}$-reset gate maps $Z$ to $\eye$, which can lower the Pauli weight. This makes the Pauli path less susceptible to noise decay. To address this, we first treat the depolarization noise as stochastically applying a complete depolarizing error or leaving the qubit unchanged, where a depolarization error has the effect of mapping non-identity Pauli operators to zero. We then reduce our problem to bounding the probability that any non-identity Pauli operator remains at the end of the adjoint circuit. This stochasticization technique was also used for Clifford circuits in \cite{nelson2024polynomialtimeclassicalsimulationnoisy}, but here, we apply it to the adjoint map of the circuit instead. Note that the action of the adjoint reset can now actually help decay the total number of remaining non-identity Pauli operators since it can map two Pauli operators to one (e.g. it maps both $\eye$ and $Z$ to $\eye$). This property that the number of remaining Pauli operators can only decrease throughout the adjoint circuit is crucial to our analysis. 

In our proof, we consider the adjoint circuit in batches of $d'$ layers at a time. Our proof proceeds inductively by alternating two steps. First, we show that all Pauli paths above a given weight $w_0$ are mapped to zero by a batch of $d'$ layers with high probability. Then we bound the number of remaining Pauli operators by using the property that each input Pauli operator to the batch of layers is mapped to at most one Pauli operator at the output. We then iterate these steps setting $w_1 = w_0/4$ and so on. It turns out that in order to achieve a sufficient decay, it is required that $d' = O(\log n)$ assuming a constant noise strength. Since we must repeat for $O(\log n)$ rounds to decay all possible Pauli weights, we get that no Pauli operators remain at the end of the adjoint circuit after $d = O(\log^2 n)$ total layers.

\subsection{Open Questions}
Prior work has shown that when circuits have no mechanism for expelling entropy, this entropy accumulates and eventually consumes the circuit \cite{aharonov1996limitationsnoisyreversiblecomputation,M_ller_Hermes_2016}. In our work, we analyze the effects of noise in the setting where the circuit is able to pump entropy out of the system by utilizing reset gates. It would be interesting to consider other mechanisms in which the circuit could reduce its entropy. The closest next step would be to generalize our results to show that worst-case Clifford circuits with amplitude damping noise must have short-term memory. Similar to the case considered in this paper, there are known fault-tolerant protocols using non-Clifford gates to preserve quantum information under amplitude damping noise \cite{shtanko2024complexitylocalquantumcircuits,benor2013quantumrefrigerator}. One reason why this could be a tractable problem to solve is that amplitude damping noise acts in a similar manner to a `random reset' error, albeit with some subtle differences. 
On the other hand, amplitude damping noise also has a unique feature in that it can act as a coherent, but noisy, rotation on stabilizer states \cite{trigueros2025nonstabilizernesserrorresiliencenoisy}; however, it is unclear whether this mechanism can be used as a mid-circuit non-Clifford gate. 

Another promising direction is to show that noisy monitored random circuits also have short-term memory. It is known that in the noiseless case, these circuits exhibit a measurement-induced phase transition \cite{PhysRevB.100.134306,PhysRevX.9.031009,PhysRevB.98.205136} where above a critical measurement rate, the output of the circuit converges to a single pure state regardless of the input state. On the other hand, below this measurement rate, the circuit preserves an extensive logical codespace for exponentially long times \cite{Gullans_2020}. When noise is introduced, it has been argued using a heuristic mapping of the circuit to a classical Ising model that this phase transition disappears and the circuit loses memory for any constant noise strength and measurement rate ~\cite{BAO2021168618, Noise_bulk, jian2021quantum, PhysRevB.108.104310, PhysRevB.107.014307, PhysRevB.108.104310, PhysRevB.107.L201113,PhysRevLett.132.240402}. It would be interesting to determine whether our techniques can be used to give a rigorous proof of this statement for the case of random Clifford circuits. Notice that the assumption of randomness is now required since there are known fault-tolerant quantum memory protocols using Clifford gates and intermediate measurements.

Additionally, it would be desirable to lift any of these statements to the case of Haar random circuits. In particular, showing that Haar random circuits under amplitude damping noise have short-term memory would strengthen the result of \cite{mele2024noiseinducedshallowcircuitsabsence} to the case of trace distance rather than observable estimation which would be a significant improvement. We also point out concurrent work \cite{SuunSoumik}, which gives further evidence for this conjecture. Our hope is that the equivalence of second-moment quantities between Clifford circuits and Haar-random circuits may provide one avenue to make progress. 

Finally, it is an interesting future direction to consider whether our no-go results on fault-tolerant quantum memory can be lifted to no-go results on computational hardness in certain noise regimes. For example, can we use these techniques to classically simulate Clifford-magic circuits with non-unital noise, similar to \cite{nelson2024polynomialtimeclassicalsimulationnoisy}? If we can also extend our results to Haar-random circuits as suggested earlier, we may be able to make progress towards classically simulating noisy Haar-random circuits with non-unital noise, which currently presents a significant gap in the literature \cite{ghosh_non-unital}.

\section{Background and Notation}

\subsection{Pauli and Clifford Group}

The Pauli group is defined as:
\begin{align*}
        \mathsf{P}_n := \{1,i,-1,-i\} \times \{\eye,X,Y,Z\}^{\otimes n}
    \end{align*}
We define the phaseless Pauli group $\hat{\mathsf{P}}_n \leqslant \mathsf{P}_n$ as:
    \begin{align*}
        \hat{\mathsf{P}}_n := \{1\} \times \{\eye,X,Y,Z\}^{\otimes n}
    \end{align*}   
    For any $O$ that is proportional to a Pauli operator $P \in \mathsf{P}_n $ let $|O|$ denote the weight of the Pauli operator $P$, i.e. the number of non-identity qubits.
    
    The Clifford group is defined as:
    \begin{align*}
        \mathsf{C}_n := \{ U \in \mathsf{U}(2^n): UPU^\dagger \in \mathsf{P}_n \forall P \in \mathsf{P}_n\}
    \end{align*}
    where $\mathsf{U}(2^n)$ is the group of $2^n \times 2^n$ unitary matrices.

\subsection{Stochastic Noise}
We define the depolarizing noise channel with noise strength $\gamma$ as:
\begin{align}
    \mathcal N_{\gamma}(\rho)  = (1-\gamma)\rho + \gamma  \frac{\eye}{2} \Tr\rho
\end{align}
This channel can be viewed as stochastically applying a complete depolarization error with probability $\gamma$, which traces out the given qubit and replaces it with the maximally mixed state. If no error occurs then the channel does nothing. Denoting this depolarization error as $\mc D$, we have that $\mathcal N_{\gamma} := (1-\gamma)\eye + \gamma  \mc D$. For a given noisy circuit $\Phi$ with mid-circuit depolarizing channels, let $b$ be a bitstring that contains a bit for each depolarizing channel representing
whether or not $\mathcal D$ is applied. We then let $\Phi_b$ represent the circuit where each noise channel is replaced with either $\eye$ or $\mc D$ as specified by $b$. Sampling $b$ according to a binomial distribution parameterized by $\gamma$, we have that by definition
\begin{align}
    \Phi = \mathbf{E}_{b}\Phi_b
\end{align}
\subsection{Reset Channel}

We define the qubit reset channel to an arbitrary state $\rho$ as follows:
\begin{align}
    \mc R_{\rho}(\sigma) = \rho \Tr\sigma
\end{align}

\subsection{Adjoint Maps}

In our analysis, it will be convenient to use the adjoint of the noisy circuit. For a given channel $\mc N$, the adjoint map $\mc N^\dagger$ is defined as follows \cite{watrous,markwilde}.
\begin{align}
    \langle X, \mathcal N(Y) \rangle = \langle \mathcal N^\dagger(X), Y\rangle
\end{align}
where $\langle X, Y \rangle := \Tr X^\dagger Y$ and $X$ and $Y$ are linear operators on the channel's Hilbert space. From this definition, it can be shown that the adjoint of a given channel $\mathcal{N}(\rho) = \sum_i K_i \rho K_i^\dagger $ is represented by $\mathcal{N}^\dagger(\rho) = \sum_i K_i^\dagger \rho K_i$. Furthermore, it is also true that given the Pauli transfer matrix of a channel, $T(\mc N)$ we have that $T(\mc N^\dagger) = T(\mc N)^\top$.

It is next helpful to explicitly characterize the adjoint map for completely depolarizing errors and qubit resets. First, a depolarizing error $\mc D$ maps the Pauli matrices as follows
\begin{align}
\begin{split}
    \mc D(I) &= I\\
    \mc D(X) &= 0\\
    \mc D(Y) &= 0\\
    \mc D(Z) &= 0
    \end{split}
\end{align}
Taking the transpose of the associated Pauli transfer matrix we have that $\mc D^\dagger = \mc D$.

Next the qubit reset channel maps the Pauli matrices as:
\begin{align}
\label{eq:resetptm}
    \begin{split}
    \mc R_\rho(I) &= 2\rho\\
    \mc R_\rho(X) &= 0\\
    \mc R_\rho(Y) &= 0\\
    \mc R_\rho(Z) &= 0
    \end{split}
\end{align}
For the adjoint of the reset channel, we can similarly take the transpose of the Pauli matrix associated with \Cref{eq:resetptm}. Suppose $\rho = \frac{I + \alpha X + \beta Y + \gamma Z}{2}$, then:
\begin{align}
    \mc R_\rho^\dagger(I) &= I \\
    \mc R_\rho^\dagger(X) &= \alpha I \\
    \mc R_\rho^\dagger(Y) &= \beta I\\
    \mc R_\rho^\dagger(Z) &= \gamma I
\end{align}
Finally, note that the adjoint of each Clifford gate is also a Clifford gate.

\section{Analysis}
\label{sec:analysis}
The primary property that we will use in our analysis is that each of the adjoint maps that make up $\Phi_b^\dagger$ is many-to-one on the set of Pauli operators since each input Pauli $s_0$ is mapped to at most one output Pauli $s_1$, i.e. $\Phi_b^\dagger(s_0) \propto s_1 \text{ or } 0$. In fact, our results apply generally to all noisy circuits whose adjoint maps have this many-to-one property.

\subsection{Reduction to Pauli Path Survival Probability}
We begin by analyzing our desired quantity in the Pauli basis, which allows us to reduce our problem to bounding the probability that a non-identity Pauli operator is output by the adjoint of the circuit. 
This is captured in the following lemma,
\begin{lemma} \label{lemma:survival_probability}
    For $\Phi,\rho,\sigma$ as in \Cref{theorem:main},
    \begin{align}
    \|\Phi(\rho) - \Phi(\sigma)\|_1 \leq 2 \Pr(\exists s \in \hat{\mathsf{P}}_n/\eye: \Phi^\dagger_b(s) \in \hat{\mathsf{P}}_n/\eye)
\end{align}
\end{lemma}
\begin{proof}
For any error configuration $b$, note the following fact. 
\begin{align}
    &\Phi_b(\rho) = \Phi_b(\sigma) \\
    &\iff \forall s \in \hat{\mathsf{P}}_n: \Tr(\Phi_b(\rho)s) = \Tr(\Phi_b(\sigma)s)\\
    &\iff \forall s \in \hat{\mathsf{P}}_n: \Tr(\rho \Phi_b^\dagger (s)) = \Tr(\sigma \Phi_b^\dagger (s))
\end{align}
The above property follows from the fact that all density matrices have a unique decomposition in the Pauli basis. Now,
\begin{align}
     \|\Phi(\rho) - \Phi(\sigma)\|_1 &= \|\mathbf{E}_b \Phi_b(\rho) - \Phi_b(\sigma)\|_1\\
     &\leq \mathbf{E}_b\|\Phi_b(\rho) - \Phi_b(\sigma)\|_1\\
     &\leq \Pr(\Phi_b(\rho) = \Phi_b(\sigma))  \cdot 0 + \Pr(\Phi_b(\rho) \neq \Phi_b(\sigma)) \cdot 2\label{eq:bigstep}\\
     &= 2\Pr(\Phi_b(\rho) \neq \Phi_b(\sigma))\\
     &\leq 2\Pr(\exists s \in \hat{\mathsf{P}}_n: \Tr(\rho \Phi_b^\dagger (s)) \neq \Tr(\sigma \Phi_b^\dagger (s)))\\
     &\leq 2\Pr(\exists s \in \hat{\mathsf{P}}_n/\eye: \Phi^\dagger_b(s) \in \hat{\mathsf{P}}_n/\eye)
\end{align}
In the last step, notice that $\Tr(\rho \eye) = \Tr(\sigma \eye)= 1$ and so the only way for $\Tr(\rho \Phi_b^\dagger (s)) \neq \Tr(\sigma \Phi_b^\dagger (s))$ is if $\Phi_b^\dagger (s)$ maps to a non-identity Pauli operator. A crucial step of this proof is \Cref{eq:bigstep}, where we have used the fact that for any choice of $b$, the trace distance $\|\Phi_b(\rho) - \Phi_b(\sigma)\|_1$ is always bounded by two. This is true because $\Phi_b$ is a quantum channel and the trace distance between two density matrices is at most two.
\end{proof}
 It now remains to upper bound $\Pr(\exists s \in \hat{\mathsf{P}}_n/\eye: \Phi^\dagger_b(s) \in \hat{\mathsf{P}}_n/\eye)$. 

\subsection{Grouping Pauli Operators based on Weight}
Our argument will rely on a careful accounting of the non-identity Pauli operators that remain at various points in the adjoint of the circuit. The probability that a Pauli operator is mapped to zero is closely related to its weight since the depolarizing channel maps non-identity Pauli operators to zero with constant probability. For a circuit $\Phi$, it is useful to divide a given set of Pauli operators $G$ into subsets based on their minimum weight throughout the circuit. These subsets are formally defined as follows
\begin{definition}
    Let $S_w(\Phi, G) \subseteq \hat{\mathsf{P}}_n$ be the set:
    \begin{align}
        S_w(\Phi, G):= \{s \in G : \min_i | \Phi_{i \leftarrow}(s)| = w\}.
    \end{align}
     where $\Phi_{i \leftarrow}$ denotes the first $i$ layers of the circuit. Here, each layer constitutes a round of disjoint 2-qubit gates, a round of depolarizing channels (or sampled depolarizing errors) and a round of arbitrarily placed qubit resets. When $G$ is not defined we take it to be $\hat{\mathsf{P}}_n$. We define $S_{\geq w}(\Phi, G)$ and $S_{<  w}(\Phi, G)$ in a similar manner.
\end{definition}
A key fact we will use is the following,
\begin{fact}
    For all $s \in S_w(\Phi, G)$ where $\Phi$ is of depth $d$, $\Pr(\Phi_b(s) \neq 0) \leq (1-\gamma)^{wd}$
\end{fact}
This is simply due to the fact that each layer of depolarizing channels will encounter at least $w$ non-identity qubits since this is the minimum weight throughout each layer of the circuit. The Pauli operator survives each depolarizing channel acting on a non-identity with probability $1-\gamma$. Because there are $d$ layers, this must occur at least $wd$ times, and so the probability of surviving all depolarizing errors is $(1-\gamma)^{wd}$.

\subsection{Batched Counting Argument}
Our strategy will be to consider $d' = O(\gamma^{-1}\log n)$ layers at a time and show that with high probability the number of remaining Pauli operators decreases substantially after each batch of layers.

\begin{lemma} \label{lemma:counting}
    For $\Phi$ as in \Cref{theorem:main} containing $d > d^*$ layers where \mbox{$d^* =  O(\gamma^{-1}\log(n) \log (n/\epsilon) )$},
    \begin{align}
    \Pr(\exists s \in \hat{\mathsf{P}}_n/\eye: \Phi^\dagger_b(s) \in \hat{\mathsf{P}}_n/\eye) \leq \epsilon
\end{align}
\end{lemma}
Note that \Cref{theorem:main} is implied by the combination of \Cref{lemma:survival_probability} and \Cref{lemma:counting}.
\begin{proof}
    Our proof proceeds by first breaking up the adjoint of the circuit into batches of $d'$ layers. Let $\Phi_{b,j}^\dagger$ be the $j$th batch of $d'$ layers of the adjoint of the circuit.
    For each $j$, we will associate a weight limit $w_j$, where $w_{j} = \ceil{w_{j-1}/4}$ and $w_0=n$. We will then define the sequences of sets $F_0, \ldots F_{j-1}, F_j$ and $G_0, \ldots G_{j-1}, G_j$ inductively as follows,
    \begin{align}
        F_0 &= G_0 = \{\hat{\mathsf{P}}_n\} \\
       F_j &= \{ \Phi_{b,j}^\dagger(s) \text{ } \forall s \in  F_{j-1} \text{ s.t. }\Phi_{b,j}^\dagger(s) \neq 0\}\\
       G_j &= \{ \Phi_{b,j}^\dagger(s) \text{ }\forall s \in S_{< w_j}(\Phi^\dagger_{b,j}, G_{j-1}) \text{ s.t. }\Phi_{b,j}^\dagger(s) \neq 0\}\}
    \end{align}
    In other words, $F_j$ and $G_{j}$ are defined as the sets of all \text{output} Pauli operators of $\Phi_{b,j}^\dagger$ starting from the \text{input} Pauli operators in $F_{j-1}$ and $G_{j-1}$ respectively, where $G_{j}$ has the added restriction that the minimum weight of the Pauli operator as it passes through $\Phi_{b,j}^\dagger$ is less than $w_j$. With this notation, our new goal is to prove that $\Pr(F_j \neq \{\eye\}) \leq \epsilon $ for some large enough $j$.
    
    Our inductive hypothesis is that $\Pr(F_{j-1} \neq  G_{j-1}) \leq (j-1) \epsilon'$ and that  $|G_{j-1}| \leq d' (9n)^{w_{j-1}} $. This is clearly satisfied in the base case of $j-1 = 0$.
    Our inductive step will be to show that $\Pr(F_j \neq G_j) \leq j' \epsilon $, and that  $|G_{j}| \leq d' (9n)^{w_{j}} $. We will iterate this argument until $w_{j} \leq 1$, which occurs at $j = O(\log n)$, at which point $G_j = \{\eye\}$. At this point,  by induction, $\Pr(F_j \neq \{\eye\}) \leq O(\epsilon'\log n) $, and so choosing $\epsilon' = O(\epsilon /\log n)$ completes the proof.

    First, we upper bound $\Pr(F_j \neq G_j)$,
        \begin{align}
        \Pr(F_j \neq G_j) &= \Pr(F_j \neq G_j | F_{j-1} = G_{j-1})\Pr(F_{j-1} = G_{j-1})  \\
        &+ \Pr(F_j \neq G_j | F_{j-1} \neq G_{j-1})\Pr(F_{j-1} \neq G_{j-1}) \\
        &\leq \Pr(F_j \neq G_j | F_{j-1} = G_{j-1})+ (j-1)\epsilon'\tag{inductive hypothesis}\\
        &\leq \Pr(\exists s \in S_{\geq w_{j}}(\Phi^\dagger_{b,j}, G_{j-1}) : \Phi^\dagger_{b,j}(s) \neq 0)+ (j-1)\epsilon' \\
        &\leq  \sum_{s \in  S_{ \geq w_{j}}(\Phi^\dagger_{b,j}, G_{j-1})} \Pr(\Phi^\dagger_{b,j}(s) \neq 0) + (j-1)\epsilon' \tag{union bound}\\
        &\leq |G_{j-1}| (1-\gamma)^{w_{j}d'} + (j-1)\epsilon'\\
        &\leq d' (9n)^{w_{j-1}}(1-\gamma)^{w_{j-1}d'/4}+ (j-1)\epsilon' \tag{inductive hypothesis}\\
        &\leq \epsilon'+ (j-1)\epsilon' \tag{setting $d' = O(\gamma^{-1} \log (n/\epsilon'))$}\\
        &= j \epsilon'
    \end{align}
    
    Next, we upper bound $|G_j|$. Consider every possible \textit{intermediate} Pauli operator at a given layer of $\Phi^\dagger_{b,j}$, whose weight is less than $w_j$. We can exploit the property that each such Pauli operator, whichever layer it is in, gets mapped by $\Phi^\dagger_{b,j}$ to at most one output Pauli operator, which may or may not be in $G_j$. Therefore, we can upper bound $|G_j|$ using the size of this set. There are $d'$ layers, $w_j$ possible values of the minimum weight, at most ${n \choose w_j}$ ways to choose the locations of it's non-identities \footnote{assuming $w_j \leq n/2$, which is valid since only $w_0 >n/2$ and this is the base case}, and at most $3^{w_j}$ ways to assign non-identities to $X$, $Y$, and $Z$. So we have,
    \begin{align}
        |G_{j}|
        &\leq w_{j} d' {n \choose w_{j}} 3^{w_{j}} \\
        &\leq w_{j}d' (\frac{en}{w_{j}})^{w_{j}}3^{w_{j}} \\
        &\leq d'(9n)^{w_{j}} \label{eq:counting}
    \end{align}
\end{proof}

\section*{Acknowledgements}
This material is based upon work supported by the U.S. Department of Energy, Office of Science, Accelerated Research in Quantum Computing, Fundamental Algorithmic Research toward Quantum Utility (FAR-Qu). Additional support is acknowledged from IBM, where JR perfomed a portion of this research as an intern under the guidance of Abhinav Deshpande, Kunal sharma, and Oles Shtanko.
We thank Dominik Hangleiter, Zhi-Yuan Wei, Daniel Malz and Alexey Gorshkov for helpful discussions.
This material is based upon work supported in part by the NSF QLCI award OMA2120757. This work was performed in
part at the Kavli Institute for Theoretical Physics (KITP), which is supported by grant NSF PHY-2309135. JN is supported by the National Science Foundation Graduate Research Fellowship Program under Grant No. DGE 2236417.

\newpage
\printbibliography
\end{document}